\documentclass[11pt,letterpaper]{article}
\usepackage[margin=1in]{geometry}

\usepackage{amsthm}
\usepackage{xcolor}
\usepackage{graphicx}
\usepackage{wrapfig}
\usepackage{authblk}

\newtheorem{theorem}{Theorem}
\newtheorem{lemma}{Lemma}
\newtheorem{observation}{Observation}
\theoremstyle{definition}
\newtheorem{definition}{Definition}
\newtheorem{problem}{Problem}

\newcommand{\etal}{\textit{et al}.}

\begin{document}
\title{Improving the dilation of a metric graph by adding edges}

\author{Joachim Gudmundsson and Sampson Wong \\
\small{University of Sydney, Australia \\
joachim.gudmundsson@sydney.edu.au, swon7907@sydney.edu.au}}

\date{}
\maketitle

\begin{abstract}
Most of the literature on spanners focuses on building the graph from scratch. This paper instead focuses on adding edges to improve an existing graph. A major open problem in this field is: given a graph embedded in a metric space, and a budget of~$k$ edges, which~$k$ edges do we add to produce a minimum-dilation graph? The special case where $k=1$ has been studied in the past, but no major breakthroughs have been made for $k > 1$. We provide the first positive result, an $O(k)$-approximation algorithm that runs in $O(n^3 \log n)$ time. 
\end{abstract}

\section{Introduction}

Let $G = (V,E)$ be a graph embedded in a metric space $(M,d_M)$. For every edge $(u, v) \in E$, the weight of the edge $(u,v)$ is equal to the distance~$d_M(u,v)$ between points~$u$ and~$v$ in the metric space~$M$. Let~$d_G(u,v)$ be the weight of the shortest path between~$u$ and~$v$ in the graph $G$. For any real number~$t > 1$, we call~$G$ a \emph{$t$-spanner} if $d_G(u,v) \leq t \cdot d_M(u,v)$ for every pair of points $u, v \in V$. The \emph{stretch}, or \emph{dilation}, of $G$ is the smallest $t$ for which $G$ is a $t$-spanner.

Spanners have been studied extensively in the literature, especially in the geometric setting. Given a fixed $t >1$, a fixed dimension $d \geq 1$, and a set of $n$ points~$V$ in $d$-dimensional Euclidean space, there is a $t$-spanner on the point set~$V$ with~$O(n)$ edges. For a summary of the considerable research on geometric spanners, see the surveys~\cite{DBLP:books/el/00/Eppstein00,DBLP:reference/crc/GudmundssonK07,DBLP:books/el/00/Smid00} and the book by Narasimhan and Smid~\cite{DBLP:books/daglib/0017763}. Spanners in doubling metrics~\cite{DBLP:conf/soda/ChanGMZ05,DBLP:conf/focs/Gottlieb15,DBLP:journals/siamcomp/Har-PeledM06} and in general graphs~\cite{DBLP:journals/rsa/BaswanaS07,peleg2000distributed,DBLP:journals/jacm/ThorupZ05} have also received considerable attention.

Most of the literature on spanners focuses on building the graph from scratch. This paper instead focuses on adding edges to improve an existing graph. Applications where graph networks tend to be better connected over time include road, rail, electric and communication networks. The overall quality of these networks depends on both the quality of the initial design and the quality of the additions. In this paper, we focus on the latter. In particular, given an initial metric graph, and a budget of~$k$ edges, which~$k$ edges do we add to produce a minimum-dilation graph?

\begin{figure}[ht]
\centering
    \includegraphics{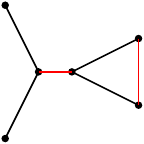}
    \caption{An example where $k=2$ edges (red) are added to an initial graph $G$ (black) to produce a minimum-dilation graph.}
    \label{fig:problem1}
\end{figure}

\begin{problem}
\label{problem:1}
Given a positive integer~$k$ and a metric graph~$G = (V,E)$, compute a set $S \subseteq V \times V$ of~$k$ edges so that the dilation of the resulting graph~$G' = (V, E \cup S)$ is minimised.
\end{problem}

The problem stated is a major open problem in the field~\cite{DBLP:journals/siamcomp/FarshiGG08,DBLP:conf/isaac/LuoW08,DBLP:journals/comgeo/Wulff-Nilsen10}. It is also one of twelve open problems posed in the final chapter of Narasimhan and Smid's book~\cite{DBLP:books/daglib/0017763}. As no major breakthroughs have been made, special cases have been studied.

The first special case is when $k=1$. Let $n$ and $m$ be the number of vertices and edges of the graph~$G$, respectively. Farshi~\etal~\cite{DBLP:journals/siamcomp/FarshiGG08} provided an $O(n^4)$ time exact algorithm and an $O(mn + n^2 \log n)$ time 3-approximation. Wulff-Nilsen~\cite{DBLP:journals/comgeo/Wulff-Nilsen10} improved the running time of the exact algorithm to $O(n^3 \log n)$, and in a follow-up paper Luo and Wulff-Nilsen~\cite{DBLP:conf/isaac/LuoW08} provided an $O((n^4 \log n)/ \sqrt m)$ time exact algorithm that uses linear space. Several of the papers that study the $k=1$ case mention the $k > 1$ case as one of the main open problems in the field. 

The second special case is if $G$ is an empty graph. Giannopoulos~\etal~\cite{DBLP:journals/ijcga/GiannopoulosKKKM10} and Gudmundsson and Smid~\cite{DBLP:journals/ijfcs/GudmundssonS09} independently proved that it is \textsc{NP}-hard to produce the highest quality spanner by adding $k$ edges to an empty graph. This implies that Problem~\ref{problem:1} is \textsc{NP}-hard. If we restrict ourselves to polynomial time algorithms, it therefore makes sense to consider approximation algorithms. 
In Euclidean space, Aronov~\etal~\cite{DBLP:journals/comgeo/AronovBCGHSV08} showed how to add $k=n-1+\ell$ edges to an empty graph to produce an $O(n/(\ell + 1))$-spanner in $O(n \log n)$ time. By setting $\ell = \varepsilon n$, this result implies an $O(1/\varepsilon)$-approximation to Problem~\ref{problem:1} for all $k \geq (1+\varepsilon) n$. However, the general case where $G$ is a non-empty (Euclidean or metric) graph and $k \leq n-1$ still remains open.

Farshi~\etal~\cite{DBLP:journals/siamcomp/FarshiGG08} conjectured that generalising their algorithm to general $k$ may provide a reasonable approximation algorithm. In Section~\ref{sec:bottleneck}, we show an~$\Omega(2^{k})$ lower bound for their algorithm. 

In this paper we obtain the first positive result for the general case. Our approximation algorithm runs in $O(n^3 \log n)$ time and guarantees an $O(k)$-approximation factor. Although our algorithm may not be optimal, we hope that we provide some insight for further research, or for related graph augmentation problems~\cite{DBLP:journals/ijcga/AhnFKSW10,DBLP:conf/giscience/HaunertM16,hurtado2013plane,DBLP:journals/comgeo/KleinKNS09}.  

We provide a tight analysis of our algorithm. We show that, for any $\varepsilon > 0$, our algorithm yields an approximation factor of $(1+\varepsilon)(k+1)$, but the same algorithm cannot yield an approximation factor better than $(1-\varepsilon)(k+1)$. We achieve our main result by reducing Problem~\ref{problem:1} to the following approximate decision version:

\begin{problem}
\label{problem:2}
Given an integer $k$, a real number $t$, and a metric graph $G = (V,E)$, decide whether $t^* \leq t$ or $t^* > \frac t {k+1}$, where $t^*$ is the minimum dilation of $G' = (V, E \cup S)$ over all sets $S$ where $S \subseteq V \times V$ and $|S| = k$. In the case where $\frac t {k+1} < t^* \leq t$, either of the two options may be chosen arbitrarily. 
\end{problem}

Our algorithm for Problem~\ref{problem:2} is a slight modification of the standard greedy $t$-spanner algorithm. 
We provide details of our algorithm and argue its correctness in Section~\ref{sec:greedy_construction}. In Section~\ref{sec:minimising_the_dilation}, we show how to use the approximate decision algorithm for Problem~\ref{problem:2} to develop an approximation algorithm for  Problem~\ref{problem:1}. We prove that only $O(\log n)$ calls to the greedy algorithm is required to obtain an $(1+\varepsilon)(k+1)$-approximation. Finally, in Section~\ref{sec:greedy_lower_bound}, we provide a construction to show that the same algorithm cannot yield an approximation factor better than $(1-\varepsilon)(k+1)$.

\section{The Greedy Construction}
\label{sec:greedy_construction}

As mentioned in the introduction, our approach to solving Problem~\ref{problem:2} is a modified greedy $t$-spanner construction. We introduce some notation for the purposes of stating the algorithm. For an edge $e \in V \times V$, let~$d_M(e)$ denote the length of the edge $e$ in the metric space $M$. Given a graph $G$, let~$\delta_G(e)$ denote the shortest path between the endpoints of $e$ in the graph $G$. Let $d_G(e)$ be the total length of edges along the path $\delta_G(e)$.

In the original greedy spanner construction, the algorithm begins with an empty graph $G$, and a positive real value $t>1$, and yields a $t$-spanner as follows: sort all the edges in $V \times V$ by increasing weight and then process them in order. Processing an edge $e$ entails a shortest path query. If $d_G(e) > t \cdot d_M(e)$, then the edge $e$ is added to $G$, otherwise it is discarded. The algorithm terminates when all edges have been processed. The resulting graph is a~$t$-spanner.

In our setting we will start with an initial graph $G$, a positive real value $t>1$ and a positive integer $k$. Our modified greedy algorithm sorts the edges in $V \times V \setminus E$ by increasing weight and then processes them in order. For each edge $e$, we perform a shortest path query. If $d_G(e) > t \cdot d_M(e)$, then the edge $e$ is added to $G$, otherwise it is discarded. The algorithm terminates if all edges have been processed, or if $k+1$ edges have been added to $G$ by the algorithm.

Formally, the greedy edges $a_i$ and the augmented graphs $G_i$ are defined inductively as follows:

\begin{definition}
\label{defn:greedy_algorithm}
Let $G_0 = G$, and for $1\leq i \leq k+1$, let $G_i = G_{i-1} \cup \{a_{i}\}$ where $a_i$ is the shortest edge in $V \times V$ satisfying $d_{G_{i-1}}(a_i) > t \cdot d_M(a_i)$.
\end{definition}

If the algorithm terminates after all the edges have been processed, then at most~$k$ edges have been added to yield a $t$-spanner. Therefore~$t^* \leq t$. Otherwise, if at least~$k+1$ edges are added, we will prove in Section~\ref{sec:greedy_upper_bound} that~$t^* > \frac t {k+1}$.  

\subsection{Proof of correctness}
\label{sec:greedy_upper_bound}

Our approach is to use the edges added by the greedy algorithm to obtain an upper bound on~$t$ with respect to $t^*$. Our upper bound comes from the following relationship, which is a straightforward consequence of Definition~\ref{defn:greedy_algorithm}:

\begin{observation}
\label{obs:L}
In the graph $G_{i-1}$, if there is a path between the endpoints of $a_i$ with total length $L$, then $L > t \cdot d_M(a_i)$.
\end{observation}

Our goal is to construct a path in $G_{i-1}$ between the endpoints of~$a_i$ and to bound its length by $(k+1)\, t^* \cdot d_M(a_i)$. If we are able to do this, then Observation~\ref{obs:L} would immediately imply that $(k+1) \, t^* > t$, as required. Note that~$i$ is some fixed integer between 1 and $k+1$. As part of our construction, we will show how to select a suitable value for~$i$.

To motivate how we construct a path in~$G_{i-1}$ between the endpoints of~$a_i$, let us consider a special case where $k=1$. Let $G$ be the initial graph and let the first two greedy edges be~$a_1$ and~$a_2$. Suppose that an optimal edge to add is $s_1$, and let $G^* = G \cup \{s_1\}$. See Figure~\ref{fig:k=1}. 

\begin{figure}[ht]
    \centering
    \includegraphics{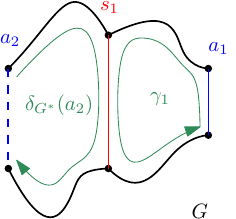}
    \caption{The graph $G$ with optimal edge $s_1$ and greedy edges $a_1$ and $a_2$.}
    \label{fig:k=1}
\end{figure}

We select $i=2$ in Observation~\ref{obs:L}, so that our goal is to construct a path in $G_1 = G \cup \{a_1\}$ between the endpoints of $a_2$ and upper bound its length by~$2\, t^* \cdot d_M(a_2)$.

A na\"ive path between the endpoints of~$a_2$ that has length upper bounded by $t^* \cdot d_M(a_2)$ is the path $\delta_{G^*}(a_2)$, which we recall is the shortest path between the endpoints of $a_2$ in the graph~$G^*$. The path is shown in Figure~\ref{fig:k=1}. The reason that $d_{G^*}(a_2) \leq t^* \cdot d_M(a_2)$ is because the dilation of~$G^*$ is~$t^*$. Unfortunately, the issue with this path is that it uses the edge~$s_1$ and therefore is not a path in~$G_{1}$, so Observation~\ref{obs:L} does not apply.

We modify $\delta_{G^*}(a_2)$ into a longer path that does not use $s_1$. Our approach is to combine the path with a cycle by using the \emph{symmetric difference} operation. Recall that the symmetric difference of a set of sets are all the elements that appear in an odd number of those sets.

To remove $s_1$ from the path~$\delta_{G^*}(a_2)$, we take its symmetric difference with the cycle~$\gamma_1$, which is formed by linking the path $\delta_{G^*}(a_1)$ and the edge $a_1$ end to end. Ideally, the symmetric difference of $\delta_{G^*}(a_2)$ and $\gamma_1$ would form a path between the endpoints of~$a_2$. Moreover, if both the path~$\delta_{G^*}(a_2)$ and the cycle~$\gamma_1$ use the edge~$s_1$ exactly once, then taking the symmetric difference cancels the two occurrences of~$s_1$, leaving a path that is entirely in~$G_1$. 

In fact, we can show this approach works in general. We begin with the na\"ive path $\delta_{G^*}(a_i)$, where~$G^*$ is the optimal graph defined as follows:

\begin{definition}
\label{defn:S}
Let $S \subseteq V \times V$ be the set of $k$ edges so that $G \cup S$ has dilation~$t^*$. Then $G^* = G \cup S$.
\end{definition}

Similar to the~$k=1$ case, the path $\delta_{G^*}(a_i)$ is not in the graph~$G_{i-1}$. We modify the path $\delta_{G^*}(a_i)$ by taking its symmetric difference with a set of cycles. We prove that for any set of cycles, the symmetric difference of $\delta_{G^*}(a_i)$ and the set of cycles always contains a path between the endpoints of $a_i$. Moreover, we show how to select the set of cycles in such a way that all edges in~$S$ are cancelled out by the symmetric difference. In this way, we have constructed a path in the graph $G_{i-1}$ between the endpoints of $a_i$.

We first prove that taking the symmetric difference of $\delta_{G^*}(a_i)$ with any set of cycles maintains the invariant that there always exists a path between the endpoints of $a_i$.

\begin{lemma}
\label{lem:symmetric_difference}
In any graph, the symmetric difference of a path $P$ with any number of cycles contains a path between the endpoints of $P$. See Figure~\ref{fig:symmetric_difference}.
\end{lemma}

\begin{figure}[ht]
    \centering
    \includegraphics{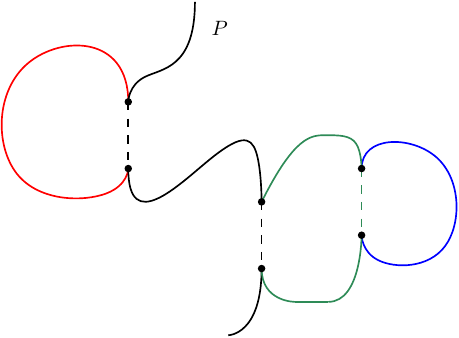}
    \caption{Given a path (black) and cycles (red, green, blue), the symmetric difference (solid) contains a path between the endpoints of the black path.}
    \label{fig:symmetric_difference}
\end{figure}

\begin{proof}
Consider a subgraph formed by the symmetric difference of $P$ and a set of cycles. We will look at the degree of all vertices in this subgraph.

Consider the parity of the degree of each vertex. Taking the symmetric difference maintains the parity of the sum of the degrees. The contribution of a cycle to the degree of all vertices is even, whereas the contribution of $P$ to the degree of all vertices is even except for the endpoints of $P$. Hence, the only two vertices with odd degree are the endpoints of $P$. Applying Euler's theorem to the connected component that contains the endpoints of $P$, we deduce that there is an Eulerian trail between the two vertices of odd degree. Hence, there is a path between the endpoints of $P$.
\end{proof}

Next, we construct the set of cycles $\Gamma = \{\gamma_j: 1 \leq j \leq k+1\}$. We will apply Lemma~\ref{lem:symmetric_difference} to our na\"ive path $\delta_{G^*}(a_i)$ and a subset of~$\Gamma$. Each cycle $\gamma_j$ is simply a generalisation of~$\gamma_1$ from the $k=1$ case, which we recall is formed by linking the path~$d_{G^*}(a_1)$ and the edge~$a_1$ end to end.

\begin{definition}
\label{defn:gamma_ij}
Let $\gamma_j$ be a cycle formed by linking the path $\delta_{G^*}(a_j)$ and the edge $a_j$ end to end.
\end{definition}

We choose an index $i$ and a subset of $\Gamma$ in such a way so that the symmetric difference of the path $\delta_{G^*}(a_i)$ and the cycles $\gamma_j \in \Gamma$ consists only of edges in $G_{i-1}$. In other words, all edges in $S$ that belong to the path $\delta_{G^*}(a_i)$ or the cycles $\gamma_j$ cancel out in the symmetric difference. We use elementary linear algebra to provide a non-constructive proof that there exists an index $i$ and a subset of $\Gamma$ where this property holds.

\begin{lemma}
\label{lem:linear_dependent}
Let $\{a_1, a_2, \ldots, a_{k+1}\}$ be the first $k+1$ edges given in Definition~\ref{defn:greedy_algorithm}. Then there exists a non-empty subset $I \subseteq \{1, 2, \ldots, k+1\}$ so that the symmetric difference of $\{\delta_{G^*}(a_j): j \in I\}$ does not contain any edges of $S$.
\end{lemma}

\begin{proof}
Recall from Definition~\ref{defn:S} that $S$ is the set of $k$ edges so that $G \cup S$ has dilation~$t^*$. Consider $\delta_{G^*}(a_j) \cap S$, which is a subset of $S$. We can represent any subset of $S$ as an element of the vector space $\{0,1\}^S$, as each binary digit simply represents whether an element is in that subset. Take the basis $\{1_{e}: e \in S\}$ for the vector space $\{0,1\}^S$. The basis element $1_{e}$ simply represents whether the~$e^{th}$ element of $S$ is in that subset. Hence, we can expand  $\delta_{G^*}(a_j) \cap S$ into a sum of basis elements by writing  $\delta_{G^*}(a_j) \cap S = \sum \lambda_{je} 1_e$.

As there are $k+1$ subsets $\delta_{G^*}(a_j) \cap S$, their vector space expansions $\sum \lambda_{je} 1_e$ must be linearly dependent. The linear dependence equation, when taken in modulo 2, can be rearranged into the form $\sum_{j \in I} \delta_{G^*}(a_j) \cap S = 0$ for some $I \subseteq \{1, 2, \ldots, k+1\}$ and $I \neq \emptyset$. The modulo 2 equation $\sum_{j \in I} \delta_{G^*}(a_j) \cap S = 0$ directly implies that the symmetric difference of $\{\delta_{G^*}(a_j) \cap S: j \in I\}$ is empty.
\end{proof}

For the remainder of this section, let $I \subseteq \{1, 2, \ldots, k+1\}$ be the subset that satisfies the conditions of Lemma~\ref{lem:linear_dependent}, in other words, the symmetric difference of $\{\delta_{G^*}(a_j): j \in I\}$ does not contain any edges of $S$. We select the path $\delta_{G^*}(a_i)$ where $i = \max I$. Let $J = I \setminus \{i\}$ and select the subset $\Gamma' = \{\gamma_j: j \in J\}$. We construct the set of edges that is the symmetric difference of $\delta_{G^*}(a_i)$ and $\Gamma'$. This completes the construction of the required path. 

\begin{figure}[ht]
    \centering
    \includegraphics{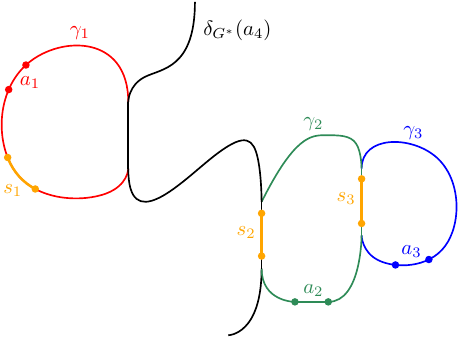}
    \caption{An example where we take the symmetric difference of $\delta_{G^*}(a_4)$, $\gamma_2$ and $\gamma_3$ to avoid all three of the edges $s_1$, $s_2$ and $s_3$ that are not in $G_3$. Note that $a_i \in \gamma_i$ for $i \in \{1,2,3\}$, and $a_4$ is the edge between the start and end points of the path $\delta_{G^*}(a_4)$.}
    \label{fig:gamma_ij}
\end{figure}

For an illustrated example, see Figure~\ref{fig:gamma_ij}. Let $k=3$ and $S = \{s_1, s_2, s_3\}$, so that $s_1 \in \delta_{G^*}(a_1)$, $s_2 \in \delta_{G^*}(a_2), \delta_{G^*}(a_4)$ and $s_3 \in \delta_{G^*}(a_2), \delta_{G^*}(a_3)$. By Lemma~\ref{lem:linear_dependent}, there must be a non-empty subset $I \subseteq \{1,2,3,4\}$ so that the symmetric difference of $\{\delta_{G^*}(a_j): j \in I\}$ does not contain any of the edges $s_1, s_2$ or~$s_3$. In particular, the subset $I = \{2,3,4\}$ includes $s_1$ zero times, and~$s_2$ and~$s_3$ both twice. Hence,  the symmetric difference of $\delta_{G^*}(a_4)$ with the cycles~$\Gamma' = \{\gamma_2,\gamma_3\}$ avoids all three of the edges $s_1,$~$s_2,$ and~$s_3$.

Now we show this symmetric difference indeed satisfies the conditions of Observation~\ref{obs:L}, so that it can be applied to yield an upper bound on~$t$ with respect to~$t^*$. Recall that the requirements of Observation~\ref{obs:L} are that the set of edges must contain a path between the endpoints of $a_i$ that uses only edges in $G_{i-1}$. By Lemma~\ref{lem:symmetric_difference}, the symmetric difference contains a path between the endpoints of~$a_i$. By Lemma~\ref{lem:linear_dependent}, we have $\{\delta_{G^*}(a_j) \cap S: j \in I\} = \emptyset$, so therefore the symmetric difference of $\{\delta_{G^*}(a_j): j \in I\}$ does not contain any edges of $S$. This implies that the symmetric difference of $\{\delta_{G^*}(a_i)\}$ and $\Gamma' = \{\gamma_j: j \in J\}$ also does not contain any edges of~$S$. Hence, we have constructed a set of edges that contains a path in~$G_{i-1}$ between the endpoints of~$a_i$, as required.

Observation~\ref{obs:L} implies an upper bound on $t$ in terms of the lengths of all the edges in the symmetric difference of $\{\delta_{G^*}(a_i)\}$ and $\Gamma' = \{\gamma_j: j \in J\}$. In Lemma~\ref{lem:upper_bound_t} we formalise this upper bound. Then, in Lemma~\ref{lem:bound_cij}, we use the fact that the dilation of~$G^*$ is~$t^*$ to obtain an upper bound on the sum of the lengths in each cycle~$\gamma_j$. In Lemma~\ref{lem:bound_cij'}, we strengthen the inequality by giving a lower bound on the length of the edges that are both in $\gamma_i$ and~$S$, and therefore cannot be part of the final symmetric difference. In Theorem~\ref{thm:problem2} we put this all together and prove the final bound~$(k+1) \, t^* > t$.

Let $c_j$ be the total length of edges in the cycle $\gamma_j$. Let $c_j'$ be the total length of edges in the intersection $\delta_{G^*}(a_j) \cap S$. Then Observation~\ref{obs:L} implies:

\begin{lemma}
\label{lem:upper_bound_t}
$d_{G^*}(a_i) + \sum_{j \in J} c_j - \sum_{j \in I} c_j' > t \cdot d_M(a_i)$
\end{lemma}

\begin{proof}
The total length of all edges in $\delta_{G^*}(a_i)$ is $d_{G^*}(a_i)$. The total length of all edges in $\gamma_j$ is $c_j$. Taking the sum $d_{G^*}(a_i) + \sum_{j \in J} c_j$ yields an upper bound on the total length of all edges in the symmetric difference $\{\delta_{G^*}(a_i)\}$ and $\{\gamma_j: j \in J\}$. However, this total length includes edges in $S$, in particular, it includes the total length of all edges in the intersections $\{\delta_{G^*}(a_j) \cap S: j \in I\}$. We know from Lemma~\ref{lem:linear_dependent} that no edge in $S$ appears in the symmetric difference, so we do not need to include any of the edges in $\{\delta_{G^*}(a_j) \cap S: j \in I\}$ in the total length. Hence, $d_{G^*}(a_i) + \sum_{j \in J} c_j - \sum_{j \in I} c_j'$ is an upper bound on the total length of the edges in the symmetric difference. Since the symmetric difference contains a path in $G_{i-1}$ between the endpoints of $a_i$, Observation~\ref{obs:L} implies the stated inequality.
\end{proof}

Next, we use the relationship between $\gamma_j$ and the graph~$G^*$ to obtain an upper bound on~$c_j$.

\begin{lemma}
\label{lem:bound_cij}
$c_j \leq (t^* + 1) \cdot d_M(a_j)$
\end{lemma}

\begin{proof}
Recall from Definition~\ref{defn:gamma_ij} that the cycle~$\gamma_j$ is the path $\delta_{G^*}(a_j)$ and the edge $a_j$ linked end to end. Since the dilation of $G^*$ is $t^*$, we have $d_{G^*}(a_j) \leq t^* \cdot d_M(a_j)$. Therefore, $c_j = d_{G^*}(a_j) + d_M(a_j) \leq (t^* + 1) \cdot d_M(a_j)$.
\end{proof}

We strengthen the inequality in Lemma~\ref{lem:upper_bound_t} by providing a lower bound on the edges that are in~$\gamma_j$ but cannot be part of the final symmetric difference. 

\begin{lemma}
\label{lem:bound_cij'}
If $t \geq (k+1) \, t^*$, then $\frac k {k+1} \cdot d_M(a_j) \leq c_j'$ for all $j$.
\end{lemma}

\begin{proof}
First, we prove the inequality 
\[
d_{G_{j-1}}(a_j) \leq d_{G^*}(a_j) + \sum_{s \in \delta_{G^*}(a_j) \cap S} d_{G_{j-1}}(s).
\]
We do so in a similar manner to Lemma~\ref{lem:upper_bound_t}. We construct a path in $G_{j-1}$ between the endpoints of $a_j$ that has length $d_{G^*}(a_j) + \sum_{s \in \delta_{G^*}(a_j) \cap S} d_{G_{j-1}}(s)$. We start with the path $d_{G^*}(a_j)$. We modify it taking the symmetric difference of $d_{G^*}(a_j)$ with a set of cycles $\beta = \{\beta_s: s \in \delta_{G^*}(a_j) \cap S\}$. The cycle $\beta_s$ is formed by linking the path $d_{G_{j-1}}(s)$ and the edge $s$ end to end. The cycle $\beta_s$ replaces every edge $s \in \delta_{G^*}(a_j) \cap S$ with the path $d_{G_{j-1}}(s) \in G_{j-1}$. Hence, the symmetric difference of $d_{G^*}(a_j)$ with the set~$\beta$ is a path in $G_{j-1}$ between the endpoints of $a_j$. Therefore, we have $d_{G_{j-1}}(a_j) \leq d_{G^*}(a_j) + \sum_{s \in \delta_{G^*}(a_j) \cap S} d_{G_{j-1}}(s)$. 

Suppose for sake of contradiction that $\frac k {k+1} \cdot d_M(a_j) > c_j'$. Consider any $s \in \delta_{G^*}(a_j) \cap S$. Then $s$ is shorter than $a_j$, since $d_M(a_j) > \frac k {k+1} \cdot d_M(a_j) > c_j' \geq d_M(s)$. In the graph $G_{j-1}$, the edge $a_j$ is a shortest edge satisfying $d_{G_{j-1}}(a_j) > t \cdot d_M(a_j)$. Since $s$ is shorter than $a_j$, we must have that $d_{G_{j-1}}(s) \leq t \cdot d_M(s)$. Now,

\[
    \begin{array}{rcl}
        d_{G_{j-1}}(a_j) 
        &\leq& d_{G^*}(a_j) + \sum_{s \in \delta_{G^*}(a_j) \cap S} d_{G_{j-1}}(s) \\
        &\leq& t^* \cdot d_M(a_j) + t \cdot \sum_{s \in \delta_{G^*}(a_j) \cap S} d_M(s) \\
        &=& t^* \cdot d_M(a_j) + t \cdot c_j' \\
        &<& t^* \cdot d_M(a_j) + t\cdot \frac k {k+1} d_M(a_j) \\
        &\leq& t \cdot \frac 1  {k+1} d_M(a_j) + t \cdot \frac k {k+1} d_M(a_j) \\
        &=& t \cdot d_M(a_j)
    \end{array}
\]
where the second last line is given by $t \geq (k+1) \, t^*$. But we know from Definition~\ref{defn:greedy_algorithm} that $d_{G_{j-1}}(a_j) > t \cdot d_M(a_j)$, so we obtain a contradiction. Therefore, we must have $\frac k {k+1} \cdot d_M(a_j) \leq c_j'$.
\end{proof}

Using Lemmas~\ref{lem:upper_bound_t}-\ref{lem:bound_cij'} we are able to prove the main result of this section.

\begin{theorem}
\label{thm:k+1_bound}
Suppose the greedy algorithm adds $k+1$ edges into the graph. Then $(k+1) \, t^* > t$.
\end{theorem}

\begin{proof}
Combining Lemmas~\ref{lem:upper_bound_t}  and~\ref{lem:bound_cij} yields:
\[
\begin{array}{rcll}
        t \cdot d_M(a_i)
        &<& d_{G^*}(a_i) + \sum_{j \in J} c_j - \sum_{j \in I} c_j' \\
        &\leq& t^* \cdot d_M(a_i) + \sum_{j \in J} (t^* + 1) d_M(a_j) - \sum_{j \in I} c_j' \\
        &=& t^* \cdot \sum_{j \in I} d_M(a_j) + \sum_{j \in J} d_M(a_j) - \sum_{j \in I} c_j'
    \end{array}
\]

Suppose for sake of contradiction that $t \geq (k+1) \, t^*$. By Lemma~\ref{lem:bound_cij'} we have $\frac k {k+1} d_M(a_j) \leq c_j'$. Summing over all $j \in I$ yields:
\[
    \begin{array}{rcl}
    \sum_{j \in I} c_j'
    &\geq& \sum_{j \in I} \frac k {k+1} d_M(a_j) \\
    &=& \frac k {k+1} d_M(a_i) + \sum_{j \in J} \frac k {k+1} d_M(a_j) \\
    &\geq& \sum_{j \in J} (\frac k {k+1} d_M(a_j) +  \frac 1 {k+1} d_M(a_i))\\
    &\geq& \sum_{j \in J} d_M(a_j)
    \end{array}
\]
The final step is because $j < i$ so $a_j$ is not longer than $a_i$. Therefore,

\[
\begin{array}{rcll}
        t \cdot d_M(a_i)
        &<& t^* \sum_{j \in I} d_M(a_j) + \sum_{j \in J} d_M(a_j) - \sum_{j \in I} c_j'\\
        &\leq& t^* \sum_{j \in I} d_M(a_j) \\
        &\leq& t^* \cdot (k+1) \cdot d_M(a_i)
    \end{array}
\]
which implies $(k+1) \, t^* > t$, as required.
\end{proof}

\subsection{Running time analysis}
\label{sec:greedy_running_time}

We analyse the running time of the greedy algorithm. Recall that the greedy algorithm sorts the edges in $\{V \times V\} \setminus E$ by increasing length and then processes them in order. Processing an edge $e$ entails a shortest path query. If $d_G(e) > t \cdot d_M(e)$, then the edge $e$ is added to $G$, otherwise it is discarded. 

Our algorithm performs efficient shortest path queries by building and maintaining an all pairs shortest paths (APSP) data structure for each of the graphs~$G_i$. When an edge $pq$ is added to the graph, the data structure updates the length of the shortest path between every pair of points $u, v \in V$. We compute the length of the three paths $u \to v$, $u \to p \to q \to v$, and $u \to q \to p \to v$, and choose the minimum length. For a fixed $u,v \in V$, this can be handled in constant time, since all pairwise distances are stored.

Hence, the overall running time of the algorithm is as follows. In preprocessing, we build the APSP data structure in $O(mn + n^2 \log n)$ time. Sorting the edges in $\{V \times V\} \setminus E$ takes $O(n^2 \log n)$ time. Querying whether $d_G(e) > t \cdot d_M(e)$ can be handled in constant time, and there are at most $O(n^2)$ such queries. Updating the APSP data structure takes $O(n^2)$ time, and there are at most $k+1$ updates. Putting this all together yields:

\begin{theorem}
\label{thm:problem2}
Given an integer $k$, a real number $t$ and a graph $G$ with $n$ vertices and $m$ edges, there is an $O((m + n\log n + kn) \cdot n)$ time algorithm that returns either YES or NO. If the algorithm returns YES, then $t^* \leq t$, otherwise, $t^* > \frac t {k+1}$.
\end{theorem}

\section{Minimising the Dilation}
\label{sec:minimising_the_dilation}

We return to Problem~\ref{problem:1}, which is to compute a $(1+\varepsilon)(k+1)$-approximation for the minimum dilation~$t^*$. For any real value~$t$, we can use Theorem~\ref{thm:problem2} to decide whether $t^* \leq t$ or $t^* > \frac t {k+1}$. Hence, it remains only to provide some bounded interval that $t^*$ is guaranteed to be in. Once we have such an interval, then we can binary search on an $\varepsilon$-grid of the interval to obtain a $(1+\varepsilon)(k+1)$-approximation.

We compute this interval in two steps. Our first step is to identify a set $T$ of $O(n^4)$ real numbers so that at least one of these numbers is an $O(n)$-approximation of~$t^*$. Our second step is to use the approximate decision algorithm in Theorem~\ref{thm:problem2} to perform a binary search on the set $T$ and yield an $O(nk^2)$-approximation for $t^*$. The $O(nk^2)$-approximation provides the required interval. 

We begin by identifying the set $T$ of $O(n^4)$ real numbers.

\begin{lemma}
\label{lem:simple_bound}
Define $T =\{\frac{d_M(u,v)} {d_M(p,q)}: u,v,p,q \in V,\, u \neq v,p \neq q\}$. Then there exists an element $t \in T$ such that $t \leq t^* \leq n \cdot t$.
\end{lemma}

\begin{proof}
Consider the graph $G^* = (V, E \cup S)$. Let the dilation of $t^*$ be attained by the pair of points $u, v \in V$. Let $pq$ be a longest edge along the shortest path from $u$ to $v$ in $G^*$. See Figure~\ref{fig:uvpq}.

\begin{figure}[tb]
    \centering
    \includegraphics{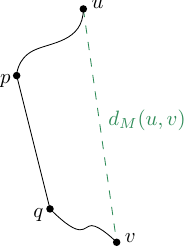}
    \caption{The edge $pq$ is a longest edge on the shortest path from $u$ to $v$.}
    \label{fig:uvpq}
\end{figure}

Recall that $d_{G^*}(u,v)$ is the length of the shortest path from $u$ to $v$ in the graph $G^*$. The dilation of $t^*$ is attained by the pair of points $u,v$, which implies $d_{G^*}(u,v) = t^* \cdot d_M(u,v)$. The shortest path from $u$ to $v$ has total length $d_{G^*}(u,v)$ and has at most $n$ edges, where the length of each edge is at most $d_M(p,q)$. This implies $d_M(p,q) \leq d_{G^*}(u,v) \leq n \cdot d_M(p,q)$. But $d_{G^*}(u,v) = t^* \cdot d_M(u,v)$, so this inequality rearranges to give $$\frac{d_M(p,q)} {d_M(u,v)} \leq t^* \leq n \cdot \frac{d_M(p,q)} {d_M(u,v)},$$
as required.
\end{proof}

Next, we use the approximate decision algorithm in Theorem~\ref{thm:problem2} to binary search the set $T =\{\frac{d_M(u,v)} {d_M(p,q)}: u,v,p,q \in V,\, u \neq v,p \neq q\}$ in order to yield an $O(nk^2)$-approximation. A na\"ive implementation of the binary search would entail computing and sorting the elements in $T$, which would require $O(n^4 \log n)$ time. To speed up our algorithm, we avoid the $O(n^4 \log n)$ preprocessing step, and we do so by using the result of Mirzaian and Arjomandi~\cite{DBLP:journals/ipl/MirzaianA85}. The result states that given two sorted lists $X$ and $Y$ each of size~$n$, one can select the~$i^{th}$ smallest element of the set $X+Y = \{x+y: x \in X,\, y \in Y\}$ in $O(n)$ time.

\begin{lemma}
\label{lem:nk^2-apx}
There is an $O((m + n\log n + kn) \cdot n \log n)$ time algorithm that computes an $O(nk^2)$-approximation for $t^*$.
\end{lemma}

\begin{proof}
In a preprocessing step, construct and sort the sets $X = \{\log(d_M(u,v)): u, v \in V,\, u \neq v\}$ and $Y = \{- \log(d_M(p,q)): p, q \in V,\, p \neq q\}$. To perform the binary search, select the $i^{th}$ smallest element of~$X+Y = \{\log(\frac {d_M(u,v)}{d_M(p,q)}): u,v,p,q \in V,\, u \neq v, p \neq q\}.$ Reverse the $\log$ transformation to obtain the $i^{th}$ smallest element of~$T$.\footnote{Alternatively, we believe it is possible to modify the algorithm of Mirzaian and Arjomandi~\cite{DBLP:journals/ipl/MirzaianA85} to select the~$i^{th}$ smallest element of the set $X/Y = \{x/y: x \in X,\, y \in Y\}$ in $O(n)$ time. If so, this may be preferred over using the log function, depending on the choice of model of computation.} Call this element~$t_i \in T$. Apply Theorem~\ref{thm:problem2} to the two dilation values $\frac 2 3 \cdot t_i$ and $n(k+1) \cdot t_i$. This returns one of three possibilities:

\begin{enumerate}
    \item $t^* \leq \frac 2 3 \cdot t_i$ and $t^* \leq n(k+1) \cdot t_i$, or
    \item $t^* > \frac 2 3 \cdot \frac {t_i} {k+1}$ and $t^* \leq n(k+1) \cdot t_i$, or
    \item $t^* > \frac 2 3 \cdot \frac {t_i} {k+1}$ and $t^* > n \cdot t_i$. 
\end{enumerate}

The fourth combination cannot occur as it yields a contradiction. Notice that in case one, we have $t^* < t_i$, so the element $t \in T$ satisfying $t \leq t^* \leq n \cdot t$ must be less than $t_i$. We can continue the binary search over the elements in $T$ that are less than $t_i$. Similarly, in case three, we have $t^* > n \cdot t_i$, so the element $t \in T$ satisfying $t \leq t^* \leq n \cdot t$ must be greater than $t_i$. We can continue the binary search over the elements in $T$ that are greater than $t_i$. In case two we halt, since we have an $O(nk^2)$-approximation for~$t^*$.

We analyse the running time of this algorithm. Sorting the sets $X$ and $Y$ takes $O(n^2 \log n)$ time. For each of the $O(\log n)$ binary search step, selecting the $i^{th}$ element of $X+Y$ takes $O(n^2)$ time~\cite{DBLP:journals/ipl/MirzaianA85}. For each of the $O(\log n)$ binary search steps, applying Theorem~\ref{thm:problem2} takes $O((m + n\log n + kn) \cdot n)$. Putting this all together yields the stated running time.
\end{proof}

Finally, we apply a multiplicative $\varepsilon$-grid to the $O(nk^2)$-approximation to yield an $(1+\varepsilon)(k+1)$-approximation. 

\begin{theorem}
\label{thm:problem1}
For any fixed $\varepsilon > 0$, there is an $O((m + n\log n + kn) \cdot n \log n)$ time algorithm that computes a $(1+\varepsilon)(k+1)$-approximation for $t^*$, where $t^*$ is the minimum dilation of $G' = (V, E \cup S)$ over all sets $S$ where $S \subseteq V \times V$ and $|S| = k$.
\end{theorem}

To simplify the running time, we note that if $k \geq n-1$, then adding the minimum spanning tree to any graph makes it an $n$-spanner, which is a $(k+1)$-approximation for the minimum dilation. Plugging in $k < n-1$ and $m \leq n^2$ into Theorem~\ref{thm:problem1} yields:

\begin{theorem}
For any fixed $\varepsilon > 0$, there is an $O(n^3 \log n)$ time algorithm that computes an $(1+\varepsilon)(k+1)$-approximation for $t^*$, where $t^*$ is the minimum dilation of $G' = (V, E \cup S)$ over all sets $S$ where $S \subseteq V \times V$ and $|S| = k$.
\end{theorem}

\section{Approximation factor no better than $(1-\varepsilon)(k+1)$}
\label{sec:greedy_lower_bound}

We provide a construction to show that the algorithms in Theorem~\ref{thm:problem2} and Theorem~\ref{thm:problem1} cannot yield an approximation factor better than $(1-\varepsilon)(k+1)$.

\begin{theorem}
For any $k \geq 1$ and $\varepsilon > 0$, there exists a graph so that for any $t \leq (1-\varepsilon)(k+1) \cdot t^*$, the greedy algorithm in Definition~\ref{defn:greedy_algorithm} adds at least $k+1$ edges to the graph.
\end{theorem}

\begin{proof}
Fix $h$ to be a small positive constant that is much smaller than $\min(\frac 1 t,\frac 1 k)$, and fix a constant~$h'$ to be arbitrarily small relative to $h$. We construct the graph~$G$ shown in Figure~\ref{fig:greedy_lower_bound}.

Let the vertices of $G$ be
\[
    \begin{array}{rcll}
    a_1 &=& (0      , 2h           ) \\
    b_i &=& (1      , 2ih          ) &\, \forall \, 1 \leq i \leq k \\
    c_i &=& (2      , 2ih          ) &\, \forall \, 1 \leq i \leq k \\
    d_i &=& (k+3+i  , 2ih          ) &\, \forall \, 1 \leq i \leq k \\
    e_i &=& (k+3+i  , (2i+1)h      ) &\, \forall \, 1 \leq i \leq k \\
    f_i &=& (2      , (2i+1)h-h'   ) &\, \forall \, 1 \leq i \leq k \\
    g_i &=& (1      , (2i+1)h      ) &\, \forall \, 1 \leq i \leq k \\
    y_1 &=& (0      , (2k+1)h      ) \\
    z_1 &=& (0      , 3h             ) \\
    \end{array}
\]

The graph $G$ is a path between these vertices. The edges of $G$ are between consecutive elements in the sequence $a_1, b_1, c_1, d_1, e_1, f_1, g_1, b_2, c_2, \ldots, f_k, g_k, y_1, z_1$. See Figure~\ref{fig:greedy_lower_bound}.

\begin{figure}[tb]
    \centering
    \includegraphics{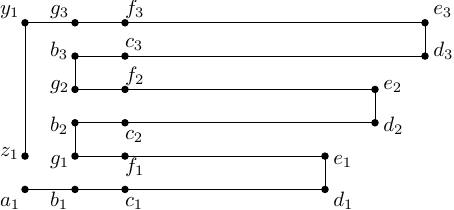}
    \caption{The construction for $k=3$.}
    \label{fig:greedy_lower_bound}
\end{figure}

The pairs of points with the largest dilation are $(a_1, z_1)$, $(b_i, g_i)$ and $(c_i, f_i)$. We can pick a small enough value of $h$ so that the dilation of all other pairs are relatively insignificant. The optimal $k$ edges to add are $(b_i, g_i)$ for all $1 \leq i \leq k$. After adding these $k$ edges, the pairs of points with the largest dilation are $(a_1, z_1)$ and $(c_i, f_i)$. Of these, the pair of points $(a_1,z_1)$ realises the maximum dilation, which is $t^* = (2 + (4k-3)h) / h \approx 2/h$, if $h$ is much smaller than $\frac 1 k$. 
 
Now let us run the greedy spanner construction for some $t \leq (1-\varepsilon)(k+1) \cdot t^*$. All pairs of points $(a_1, z_1)$, $(b_i, g_i)$ and $(f_i, c_i)$ start off with dilation greater than $2(k+2)/h$. But $2(k+2)/h = (k+2) \cdot 2/h > (k+1) \cdot t^* > t$, where the second inequality is true for sufficiently small values of~$h$. The pairs of points with highest dilation are $(a_1, z_1)$, $(b_i, g_i)$ and $(f_i, c_i)$, and the edges connecting these pairs of points satisfies $d_{G_i}(e) > t \cdot d_M(e)$. The shortest of these edges will be added first by the greedy $t$-spanner construction. The pairs $(c_i, f_i)$ have distance $h - h'$, making the edge between them the shortest and first to be considered by the greedy algorithm. Adding an edge between $(c_i, f_i)$ does not reduce the dilation of the other pairs of points $(c_j, f_j)$. Therefore, the greedy spanner construction first adds the edges $(c_i, f_i)$ for all $1 \leq i \leq k$. 

After adding $(c_i, f_i)$ for all $1 \leq i \leq k$, the dilation between the pair of points $a_1$ and $z_1$ is now $(2k+2 + (4k-3)h) / h$. But $(2k+2 + (4k-3)h) / h = (2k+2 + (4k-3)h) / (2 + (4k-3)h) \cdot t^* > (1-\varepsilon)(k+1) \cdot t^*$ for sufficiently small values of $h$ relative to $\varepsilon$. Therefore, the greedy $t$-spanner construction must add the edges $(c_i, f_i)$ for all $1 \leq i \leq k$ plus at least one additional edge, so it adds at least~$k+1$ edges in total.
\end{proof}

Our construction shows that in Theorem~\ref{thm:problem2} we cannot hope to obtain a bound that is much better than $t^* > \frac t {k+1}$. Similarly, in Theorem~\ref{thm:problem1}, our construction implies that the algorithm may continue searching for higher dilation values up until $(1-\varepsilon)(k+1) \cdot t^*$. Therefore, we cannot hope to obtain a much better approximation ratio than~$(1+\varepsilon)(k+1)$ with our algorithm.

\section{Farshi~\etal's Conjecture}
\label{sec:bottleneck}

Farshi~\etal~\cite{DBLP:journals/siamcomp/FarshiGG08} conjectured that generalising their algorithm to general $k$ may provide a reasonable approximation algorithm. We show an~$\Omega(2^{k})$ lower bound for their algorithm. 

Farshi~\etal~\cite{DBLP:journals/siamcomp/FarshiGG08} studied the special case where $k=1$. They achieved a 3-approximation by adding the \emph{bottleneck edge}, which is an edge between a pair of points that achieves the maximum dilation. They also provided a generalisation of their algorithm for $k > 1$. The generalisation consists of~$k$ stages. In each stage, the dilation of the graph is computed, and a pair of points that achieves the maximum dilation is identified. Then an edge is added between those pair of points. Formally, given an initial metric graph $G$, and an integer $k$:

\begin{definition}
\label{defn:bottleneck}
Let $G_0 = G$, and for $1 \leq i \leq k$, let $G_i = G_{i-1} \cup \{b_i\}$ where $b_i$ is an edge between the pair of points that achieves the maximum dilation of $G_{i-1}$.
\end{definition}

Farshi~\etal~\cite{DBLP:journals/siamcomp/FarshiGG08} conjectured that the dilation of the augmented graph~$G_k$ may be reasonable approximation for the dilation of the optimal graph~$G^*$. We provide a negative result that states that their algorithm cannot yield an approximation factor better than $2^k$.

\begin{theorem}
For any $k \geq 1$, there exists a initial graph $G$  where bottleneck algorithm in Definition~\ref{defn:bottleneck} yields a graph~$G_{k}$ with dilation $2^k$ times that of the dilation of the optimal graph~$G^*$.
\end{theorem}

\begin{proof}
Fix $h$ to be a small constant. Let the vertices of $G$ be 
\[
    \begin{array}{rclclll}
    x_0 &=& (-1, h) \\
    y_i &=& (0, 2^i h) &\forall \quad 1 \leq i \leq k+1\\
    z_i &=& (2^{i-1}, 3 \cdot 2^{i-1} h)&\forall \quad 1 \leq i \leq k\\
    x_1 &=& (-1, 2^{k+1}h + h)
    \end{array}
\]
Join the vertices together to form a path $x_0, y_1, z_1, y_2, z_2, \ldots, y_k, z_k, y_{k+1}, x_1$. See Figure~\ref{fig:2^k}.

\begin{figure}[tb]
    \centering
    \includegraphics{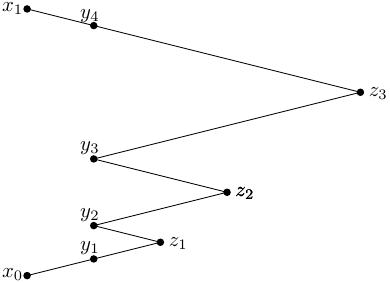}
    \caption{The construction for $k=3$.}
    \label{fig:2^k}
\end{figure}

It is straightforward to check that all edges in $G$ have gradient $\pm h$. Since $h$ is a small constant, all edges are almost horizontal. Therefore, the pairs of vertices with maximum dilation are those that are vertically above one another, in other words, the pairs $(x_0,x_1)$, or $(y_i,y_{i+1})$ for $1 \leq i \leq k$. In particular, all the pairs listed have a dilation value of $\sqrt{1+h^2}/h$.

Since $(x_0, x_1)$ is one of the pairs of vertices with maximum dilation, we can choose the first bottleneck edge $b_1$ to connect these two points. It is easy to check that since the distance in the graph between $(x_0,x_1)$ is at least twice the distance of any other pair $(y_i, y_{i+1})$, adding the first bottleneck edge does not reduce the dilation of any of the pairs $(y_i, y_{i+1})$. Inductively, we can show that for $i \geq 2$, $b_i = (y_{k-i+2}, y_{k-i+3})$ is the $i^{th}$ bottleneck edge added. This is because it initially had the maximum dilation of~$\sqrt{1+h^2}/h$, and adding the bottleneck edges $b_1, b_2, \ldots b_{i-1}$ did not reduce its dilation factor. Finally, after adding $b_1, \ldots b_k$, the dilation of the augmented graph $G_k$ is still $\sqrt{1+h^2}/h$ and is attained by $(y_1, y_2)$. 

The optimal placements of $k$ edges would be the edges $(y_1, y_2), \ldots (y_k, y_{k+1})$. Under this placement of $k$ edges, the maximum dilation value is attained by $(x_0, x_1)$, and is at least $2\sqrt{1+h^2} / (2^{k+1} \cdot h)= \sqrt{1+h^2}/(2^k \cdot h)$. Hence, the augmented graph $G_k$ has a dilation of $2^k$ times the dilation of the optimal graph $G^*$.

Note that in our construction, ties are broken adversarially when choosing the bottleneck edge to add. If we would like to lift the requirement on the adversarial choice of which bottleneck edge to add, we can perturb $x_0$ and $x_1$ vertically towards each other, which guarantees that $(x_0, x_1)$ is the first bottleneck edge to be added. We can do so similarly for the other bottleneck edges.
\end{proof}

\section{Concluding Remarks}
In Farshi~\etal~\cite{DBLP:journals/siamcomp/FarshiGG08} it was conjectured that generalising their algorithm to any positive integer $k$ may provide a reasonable approximation algorithm. In Section~\ref{sec:bottleneck}, we showed an $\Omega(2^k)$  lower bound for the approximation factor. We obtained the first positive result for the general case. Our approximation algorithm runs in $O(n^3 \log n)$ time and guarantees an $O(k)$-approximation factor. 

Two obvious open problems are to develop an algorithm with a better approximation factor, or to show an inapproximability bound.

\bibliographystyle{plain}
\bibliography{main}

\end{document}